\documentclass[12pt,reqno,letterpaper]{amsart}

\usepackage{euscript,verbatim}
\usepackage[paper=letterpaper,margin=1in]{geometry}

\newcommand{\bC}{{\mathbb C}}

\newcommand\CC{\mathbb{C}}
\newcommand\RR{\mathbb{R}}
\newcommand\ZZ{\mathbb{Z}}
\newcommand\NN{\mathbb{N}}

\newtheorem{theorem}{Theorem}[section]
\newtheorem{lemma}[theorem]{Lemma}

\newcommand\be{\begin{equation}}
\newcommand\ee{\end{equation}}
\newcommand\bes{\begin{eqnarray}}
\newcommand\ees{\end{eqnarray}}
\newcommand\ben{\begin{eqnarray*}}
\newcommand\een{\end{eqnarray*}}
\newcommand\non{\nonumber}

\newcommand\Ra{\Rightarrow}
\newcommand{\eq}[1]{Eq. \eqref{#1}}
\newcommand{\eqs}[2]{Eqs. \eqref{#1} and \eqref{#2}}

\newcommand{\sect}[1]{Section \ref{#1}}
\newcommand{\lem}[1]{Lemma~\ref{#1}}
\newcommand{\thm}[1]{Theorem~\ref{#1}}

\newcommand{\half}{\frac{1}{2}}
\newcommand{\hf}[1]{\frac{#1}{2}}
\newcommand{\inv}[1]{\frac{1}{#1}}
\newcommand\res{{\rm Res}}
\newcommand{\intab}{\int_a^b}

\newcommand\diff[2]{\frac{\textrm{d}{#1}}{\textrm{d}{#2}}}

\newcommand\diffp[2]{\frac{\partial{#1}}{\partial{#2}}}

\begin{document}

\title{Zeta Function on Surfaces of Revolution}
\author{Thalia D. Jeffres}
\address{Department of Mathematics, Statistics, and Physics\\
Wichita State University\\
1845 Fairmount, Box 33\\
Wichita, KS 67260-0033\\USA}
\email{jeffres@math.wichita.edu}

\author{Klaus Kirsten}
\address{Department of Mathematics\\ Baylor University\\
         Waco\\ TX 76798\\USA }
\email{Klaus$\_$Kirsten@baylor.edu}

\author{Tianshi Lu}
\address{Department of Mathematics, Statistics, and Physics\\
Wichita State University\\
1845 Fairmount, Box 33\\
Wichita KS 67260-0033\\USA}
\email{lu@math.wichita.edu}

\begin{abstract}
In this paper we applied the contour integral method for the zeta function associated with a differential operator to the Laplacian on a surface of revolution. Using the WKB expansion, we calculated the residues and values of the zeta function at several important points. The results agree with those obtained from the heat kernel expansion. We also obtained a closed form formula for the determinant of the Laplacian on such a surface.
\end{abstract}
\maketitle

\section{Introduction}
Spectral functions, namely functions of the spectrum of, usually, second order differential operators play an important
role in mathematics and physics. The possibly most important such spectral function is the zeta function which directly relates to topics such as
analytic torsion \cite{ray71-7-145}, the heat kernel \cite{gilk95b,vass03-388-279}, Casimir energy \cite{bord09b,bord01-353-1,cogn01-34-7311,eliz95b,milt01b}, effective
actions and quantum tunneling \cite{buch92b,byts96-266-1,cole77-15-2929,cole78-58-211,dowk76-13-3224,espo97b,hawk77-55-133,schu81b} and the study of critical metrics \cite{burg94-120-440,okik05-17-105,osgo88-80-148}. The
relationship between the zeta function and all these topics is established by the fact that very specific properties of the zeta function encode
the information needed in these areas. For example residues and values at specific points contain information about the short time heat kernel asymptotics
with all the geometric information it contains, the (finite part of its) value at $s=-1/2$ relates to the Casimir energy, and its derivative at $s=0$ is relevant in all the other subjects
mentioned above; see, e.g., \cite{kirs02b}. It is therefore no surprise that an enormous effort is made to understand spectral zeta functions as completely as possible.

The situation is particularly satisfying for the one-dimensional case where functional determinants can be obtained in closed form as boundary values of a suitable initial value problem
\cite{burg91-138-1,burg93-16-496,drey78-45-15,form87-88-447,form92-147-485,gelf60-1-48,levi77-65-299}. A contour integral approach established in
\cite{kirs03-308-502,kirs04-37-4649} was used to rederive this case and generalize it to arbitrary boundary conditions. This approach has the advantage that it allows for the systematic study of higher
dimensional configurations as long as the Laplacian separates in a suitable fashion. Examples are Laplacians with spherically symmetric potentials \cite{dunn06-39-11915,dunn09-42-075402} and
Laplacians for spherical suspensions \cite{fuccunp}. In those cases the analysis of the associated zeta function is based on explicitly known special functions, namely Bessel functions
and Legendre functions.

Another example, studied in the present paper, is the Laplacian on a compact surface of revolution.
These surfaces are geometrically interesting enough to exhibit many features that appear in the study of spectral functions and their
asymptotics, while at the same time the rotational symmetry allows the explicit calculations that are necessary in order to understand
meromorphic extensions of the zeta function to regions outside of the domain of convergence. The complication with this example is that
no known special function emerges as a solution to the relevant initial value problem because the function $f$ whose graph is used to generate the
surface of revolution is kept fairly general. As a result needed asymptotic behavior has to be determined from scratch starting from an ordinary
differential equation. The WKB method is our method of choice \cite{bend10b,mill06b} and it allows us to obtain
the values or residues of the zeta function at several important points, in addition to the determinant of the Laplacian.
The relation between the values of the zeta function and the global geometrical properties of the surface is revealed.
Furthermore it is notable that the result for the determinant is a closed-form formula.

The structure of the paper is as follows. In \sect{zeta} we present the geometric preliminaries, define the generalized zeta function,
and reformulate it as sum of contour integrals. The zeta function is written as the sum of a one-dimensional term and a two-dimensional term, which are calculated in \sect{1D} and \sect{2D} respectively.
The values and residues of the zeta function at several selected points are calculated along with the determinant of the Laplacian in \sect{2D}.
We conclude in \sect{discuss} with a comparison of our method and other approaches.

\section{Zeta function for a surface of revolution}\label{zeta}

\subsection{Laplacian on a surface of revolution}
We consider the Laplacian on a surface of revolution with Dirichlet boundary values. Let $f\in C^2(a,b)$ be a positive function from $[a,b]$ to $\RR$, and let $M$ be the compact surface with boundary that is generated by revolving the graph of $f$ around the $x$-axis. With the metric induced by the Euclidean metric of $\RR^3$, $(M,g)$ becomes a Riemannian manifold. In the natural coordinates $(x,\theta)$, the metric is
$$ \left[ \begin{array}{ll}
          1+ f'(x)^{2} & 0 \\
          0 & f^{2} (x)
          \end{array} \right] .$$
Recalling the formula for the Christoffel symbols in terms of the metric,
$$ \Gamma_{ij}^k = \half\sum _{l=1}^2 g^{kl} \left(\frac{\partial g_{il}}{\partial x_j} + \frac{\partial g_{jl}}{\partial x_i}
 - \frac{\partial g_{ij}}{\partial x_l}\right), $$
we have
$$ \Gamma_{11}^1 = \frac{f'f''}{1+ f'^2},\ \Gamma_{12}^2 = \Gamma_{21}^2 = \frac{f'}{f},\ \Gamma_{22}^1 = \frac{-ff'}{1+f'^2}, $$
and the others are zero. We then calculate that
\be\label{lap}
\Delta u = \frac{1}{1+f'^2} \left[\frac{\partial^2u}{\partial x^2} + (\frac{f'}{f}-\frac{f'f''}{1+f'^2})\frac{\partial u}{\partial x}
 + \frac{1+f'^2}{f^2} \frac{\partial^2u}{\partial\theta^2}\right].
\ee
The eigenvalue equation is $\Delta u = -\lambda u$. The following facts are standard \cite{gall04b,scho94b}:
\begin{enumerate}
\item The eigenvalues are real and positive; if labeled and ordered as $0 <\lambda_1\leq\lambda_2\leq\ldots$, then $\lambda_k\to\infty$.
\item For each $\lambda_k$, the corresponding eigenspace in $L^2(M)$ is finite-dimensional.
\item The spectrum of $M$ is discrete.
\end{enumerate}
We remark here that zero is not an eigenvalue of the Laplacian if Dirichlet boundary values are taken.

\subsection{Generalized zeta function}

Our spectral function for the Laplacian generalizes the well-known Riemann zeta function,
\be\label{rzeta}
\zeta_R(s) = \sum_{k=1}^\infty k^{-s},
\ee
which converges for $\Re(s)>1$. For the Laplacian on a surface of revolution with Dirichlet boundary condition, if the eigenvalues are labeled and ordered as $0<\lambda_1\leq \lambda_2\leq\ldots$, the generalized zeta function is defined to be
\be\label{gzeta}
\zeta(s) = \sum_{\lambda_k \in {\rm Spec}(\Delta)} \lambda_k^{-s},
\ee
for those values of $s\in\bC$ for which this sum converges. Weyl's asymptotic formula can be used to determine the domain of convergence of \eq{gzeta} \cite{chav84b}. On a compact Riemannian manifold $M$ of dimension $n$,
$$\lambda_k^{\hf{n}} \sim \frac{2^n\pi^{\hf{n}}\Gamma(\hf{n}+1)}{{\rm Vol}(M)}k,\ k\to\infty.$$
From this we determine that the zeta function is defined for complex values of $s$ having $\Re(s)>n/2$, in our case $\Re(s)>1$. However, just as for the Riemann zeta function, much information of geometric and physical significance is contained in $\zeta(s)$ for complex values of $s$ outside of this region. $\zeta(s)$ has a simple pole at $s=1$. The holomorphic continuation of $\zeta(s)$ to $\Re(s)<1$ turns out to have simple poles at $s=\half-n$, $n\in\NN$. Here we obtain expressions for the residues of $\zeta(s)$ at $s=1,\half,-\half$, and its value at $s=0$, in terms of geometrical properties of the surface. Furthermore we compute
the determinant of the Laplacian using the relation with the zeta function given by \cite{dowk76-13-3224,hawk77-55-133,ray71-7-145}
\be\label{det}
\ln\det\Delta = -\zeta'(0).
\ee
We obtain a formula for $\zeta'(0)$ given entirely in terms of the function $f$ whose graph generates the surface of revolution.

\subsection{Reformulation of the zeta function as a contour integral}

A fundamental difficulty in studying generalized zeta functions is that the eigenvalues are not known explicitly except for special cases. However, as described in Ref. \cite{kirs02b,kirs08-76-60}, the zeta function can be reformulated as a contour integral, which is highly suitable to such investigations. Here we adapt this method to the Laplacian on a surface of revolution.

The rotational symmetry of $M$ suggests the use of separation of variables to describe the eigenvalues of the Laplacian. The eigenfunctions of the Laplacian on the cross-section $S^1$, denoted by $u_k(\theta)$, $k\in\ZZ$, are $e^{ik\theta}$, with eigenvalue $-k^2$. Let $u(x,\theta) = \phi(x)u_k(\theta)$. By \eq{lap}, $u$ is an eigenfunction of $\Delta$ on $M$ with eigenvalue $\lambda$ if and only if $\phi$ is a nontrivial solution to
\be\label{phi-ode}
\phi''(x) + \left(\frac{f'}{f} - \frac{f' f''}{1+f'^2}\right)\phi'(x) + \left(\lambda-\frac{k^2}{f^2}\right)(1+f'^2)\phi(x) = 0,
\ee
satisfying the boundary conditions
$$ \phi(a)=0,\ \phi(b)=0. $$
Denote by $\phi_k(\lambda;x)$ the solution to the initial value problem consisting of the same differential equation (\ref{phi-ode}), but with initial conditions
\be\label{phi-ini}
\phi(a)=0,\ \phi'(a)=1.
\ee
Then for each integer $k$, $\lambda\in\CC$ is an eigenvalue if and only if $\phi_k(\lambda;b)=0$. For each fixed integer $k$, order the eigenvalues and denote them by $0<\lambda_{k,1}\leq\lambda_{k,2}\leq\ldots$. Since $\phi_k(\lambda;b)$, regarded as a function of $\lambda$, vanishes at and only at $\lambda_{k,n}$, $n\in\ZZ^+$, its canonical product representation \cite{ahlf79b} is
$$\phi_k(\lambda;b)=\phi_k(0;b)\prod_{n=1}^\infty \left(1-\frac{\lambda}{\lambda_{k,n}}\right).$$
The convergence of the infinite product is guaranteed by Weyl's asymptotic formula. Let
$$D_k(\lambda)=\frac{\phi_k(\lambda;b)}{\phi_k(0;b)}=\prod_{n=1}^\infty \left(1-\frac{\lambda}{\lambda_{k,n}}\right).$$
By the Residue Theorem, we may express $\zeta(s)$ as the sum of contour integrals.
\be\label{contour}
\zeta(s) = \inv{2\pi i}\int_\Gamma\lambda^{-s}\diff{\ln D_0(\lambda)}{\lambda}d\lambda
 + 2\sum_{k=1}^\infty\inv{2\pi i}\int_\Gamma\lambda^{-s}\diff{\ln D_k(\lambda)}{\lambda}d\lambda,
\ee
where $\Gamma$ is any counterclockwise contour that encloses all the $\lambda_{k,n}$'s but not the origin.
The factor 2 in the second term above comes from the fact that $\lambda_{k,n}=\lambda_{-k,n}$. Following Ref. \cite{kirs08-76-60} (Fig. 3),
the contour can be chosen to go around the negative axis, from $-\infty$ to 0 above the real axis, along which $\lambda=e^{i\pi}y$, and from 0 to $-\infty$ below the real axis, along which $\lambda=e^{-i\pi}y$. Using this contour, \eq{contour} is simplified to be
\be\label{master}
\zeta(s) = \zeta_1(s)+\zeta_2(s),
\ee
where
\be\label{zeta1}
\zeta_1(s)=\frac{\sin\pi s}{\pi}\int_0^\infty y^{-s}\diff{\ln D_0(-y)}{y}dy,
\ee
\be\label{zeta2}
\zeta_2(s)=2\frac{\sin\pi s}{\pi}\sum_{k=1}^\infty \int_0^\infty y^{-s}\diff{\ln D_k(-y)}{y}dy.
\ee
Note, that Eq. (\ref{zeta2}) should be understood in the sense that the analytical continuation of the integral is constructed first
and the summation of that result is performed afterwards.

Next, in \sect{1D}, we will use the WKB method to solve \eq{phi-ode} with $k=0$, and compute important values of $\zeta_1(s)$ and $\zeta'_1(0)$. In \sect{2D} we will adapt the method to the series over $k$, and compute $\zeta_2(s)$ along with $\zeta'_2(0)$.

\section{Zeta function and derivative with $k=0$}\label{1D}
By Weyl's asymptotic formula \eq{zeta1} converges for $\Re s>\half$. To obtain the meromorphic extension of $\zeta_1(s)$ to $\Re s<\half$, we apply the WKB method to determine the asymptotics of $\phi_0(-y;b)$; see for example \cite{bend10b,mill06b} for a discussion of the WKB method. First, we make the change of variables $y=z^2$ in the differential equation (\ref{phi-ode}) with $k=0$. It becomes
\be\label{phi1}
\phi''(x) + \left(\frac{f'}{f} -\frac{f'f''}{1+f'^2}\right)\phi'(x) - z^2(1+f'^2)\phi(x) = 0.
\ee
Denote the coefficient of $\phi'(x)$ in \eq{phi1} by $u(x)$,
$$ u(x) = \frac{f'}{f} - \frac{f'f''}{1+f'^2}.$$
Let
\be\label{Sdef}
\phi(x)=\exp\left(\int_a^x[S(z;\sigma)-\hf{u(\sigma)}]d\sigma\right).
\ee
Substituted into \eq{phi1}, the equation for $S(z;x)$ is
$$ S'(z;x)+S(z;x)^2=(1+f'^2)z^2+\frac{u^2}{4}+\hf{u'},$$
where the prime indicates differentiation with respect to $x$.
The solution can be written as asymptotic expansions in $z$. It has two branches, namely,
\be\label{Spm}
S_\pm(z;x)=\pm s_{-1}z+s_0\pm\frac{s_1}{z}+\ldots=\pm(s_{-1}z+\frac{s_1}{z}+\ldots)+(s_0+\frac{s_2}{z^2}+\ldots),
\ee
where the coefficients $s_n(x)$ can be computed recursively. The first three coefficients are
\be\label{sf1}
s_{-1}=\sqrt{1+f'^2},\ s_0=-\frac{f'f''}{2(1+f'^2)},\ s_1=-\inv{8}\frac{f'^2}{f^2\sqrt{1+f'^2}}+\inv{4}\frac{f''}{f(1+f'^2)^{3/2}}.
\ee
Applying the initial condition $\phi(a)=0$ and $\phi'(a)=1$, the solution to \eq{phi1} is
\begin{eqnarray}
\phi_0(-z^2;x)&=&e^{-\half\int_a^xu(\sigma)d\sigma}\frac{e^{\int_a^xS_+(z;\sigma)d\sigma}-e^{\int_a^xS_-(z;\sigma)d\sigma}}{S_+(z;a)-S_-(z;a)}\nonumber\\
&=&e^{\int_a^x[S_+(z;\sigma)-\half u(\sigma)]d\sigma}\frac{1-e^{-\int_a^x[S_+(z;\sigma)-S_-(z;\sigma)]d\sigma}}{S_+(z;a)-S_-(z;a)}.\nonumber\end{eqnarray}
Let $S_\pm=\pm S_1+S_2$, where
\be\label{S12}
S_1(z;x)=s_{-1}z+\frac{s_1}{z}+\cdots,\ \ \ S_2(z;x)=s_0+\frac{s_2}{z^2}+\cdots.
\ee
Notice that
$$S'_++S_+^2=S'_-+S_-^2\Ra S_1'+2S_1S_2=0\Ra \ln S_1(z;a)-\ln S_1(z;b)=2\int_a^bS_2(z;x)dx.$$
$\phi_0(-z^2;b)$ can be written as
\be\label{phi1S1}
\phi_0(-z^2;b)=e^{\int_a^b[S_1(z;x)-\half u(x)]dx}\frac{1-e^{-\int_a^b2S_1(z;x)dx}}{2\sqrt{S_1(z;a)S_1(z;b)}}.
\ee
Now we apply \eq{phi1S1} to the contour integral \eq{zeta1}. Noticing that \eq{Spm} is invalid as $z\to0$, we write the zeta function as
\be\label{zeta1z}
\zeta_1(s)=\frac{\sin\pi s}{\pi}\left[\int_0^1 y^{-s}\diff{\ln D_0(-y)}{y}dy+\int_1^\infty z^{-2s}\diff{\ln D_0(-z^2)}{z}dz\right].
\ee
By \eq{phi1S1},
\begin{eqnarray}\label{lnD1}
\ln D_0(-z^2)&=&\int_a^b\left[S_1(z;x)-\hf{u(x)}\right]dx-\hf{\ln S_1(z;a)+\ln S_1(z;b)}\nonumber\\
& &+\ln\left[1-e^{-2\int_a^bS_1(z;x)dx}\right]-\ln(2\phi_0(0;b)).
\end{eqnarray}
To compute $dD_0(-z^2)/dz$, we need the asymptotic expansion of $\partial S_1(z;x)/\partial z$, which can be derived from \eq{S12}, and the expansion of $\partial \ln S_1(z;x)/\partial z$. Let
$$\diffp{\ln S_1(z;x)}{z}=\diffp{}{z}\left[\ln(s_{-1}z)+\frac{s_1}{s_{-1}z^2}+\cdots\right]=\frac{t_0(x)}{z}+\frac{t_1(x)}{z^3}+\cdots,$$
where
\be\label{tf1}
t_0(x)=1,\ \ \ t_1(x)=-2\frac{s_1}{s_{-1}}=\inv{4}\frac{f'^2}{f^2(1+f'^2)}-\half\frac{f''}{f(1+f'^2)^2}.
\ee
Substituted into \eq{lnD1},
\ben
\frac \partial {\partial z} \ln D_0(-z^2)&=&\int_a^b s_{-1}(x)dx-\inv{z^2}\int_a^b s_1(x)dx-\cdots\\
&&-\half\frac{t_0(a)+t_0(b)}{z}-\half\frac{t_1(a)+t_1(b)}{z^3}-\cdots
\een
Now we substitute the equation above into the zeta function in \eq{zeta1z}. The last integral above is uniformly bounded for s with $\Re s$ bounded below, while the integration from 0 to 1 in \eq{zeta1z} is uniformly bounded for $s$ bounded above by $1-\epsilon$ for any $\epsilon>0$. Using the meromorphic extension of the following integrals to $s\in\CC$,
\be\label{intz1}
\int_1^\infty z^{-2s}z^{-n}dz=\inv{2s+n-1},\ \ \ n=0,1,2,\ldots
\ee
we can compute the residue of the integral from 1 to $\infty$ in \eq{zeta1z} at $s=\half,0,-\half,$ etc. If $s$ in an half integer, $\sin\pi s\neq0$, and $\zeta_1$ has a single pole at $s$ for a general $f(x)$. If $s$ is an integer, $\sin\pi s$ has a single zero at $s$, and so $\zeta_1$ has a finite value at $s$. For example, using \eq{intz1} with $n=0,1,2,3$ respectively, and plugging into \eqs{sf1}{tf1}, we have
\be\label{zeta11}
\res\zeta_1 \left(\half\right)=\inv{2\pi}\int_a^bs_{-1}(x)dx=\inv{2\pi}\int_a^b\sqrt{1+f'(x)^2}dx,
\ee
\be\label{zeta12}
\zeta_1(0)=-\frac{t_0(a)+t_0(b)}{4}=-\half,
\ee
\be\label{zeta13}
\res\zeta_1\left(-\half\right)=\inv{2\pi}\int_a^bs_1(x)dx
=\inv{2\pi}\int_a^b\left[-\inv{8}\frac{f'^2}{f^2\sqrt{1+f'^2}}+\inv{4}\frac{f''}{f(1+f'^2)^{3/2}}\right]dx,
\ee
\bes\label{zeta14}
\zeta_1(-1)&=&\frac{t_1(a)+t_1(b)}{4}=\inv{16}\left[\frac{f'(a)^2}{f(a)^2(1+f'(a)^2)}-\frac{2f''(a)}{f(a)(1+f'(a)^2)^2}\right.\non\\
&&+\left.\frac{f'(b)^2}{f(b)^2(1+f'(b)^2)}-\frac{2f''(b)}{f(b)(1+f'(b)^2)^2}\right].
\ees
The calculation can be continued for $s<-1$, but the expressions become very complicated and we do not include them here.

\vspace{0.2in}{\bf Derivative of $\zeta_1(s)$ at $s=0$}

To compute $\zeta_1'(0)$, we use integration by parts. Rewrite $\zeta_1(s)$ as
\bes\label{zeta1ip}
\zeta_1(s)&=&\frac{\sin\pi s}{\pi}s\int_0^\infty y^{-s-1}\ln D_0(-y)dy\non\\
&=&\frac{\sin\pi s}{\pi}s\left[\int_0^1 y^{-s-1}\ln D_0(-y)dy+2\int_1^\infty z^{-2s-1}\ln D_0(-z^2)dz\right].
\ees
Since $\ln D_0(-y)=O(y)$ for $y\to0$, the integration from 0 to 1 is bounded for $s=0$. For the integration from 1 to $\infty$, by the asymptotic expansion \eq{S12},
$$ S_{1}(z;x) = s_{-1}(x)z + O(z^{-1})=s_{-1}(x)z\cdot[1+O(z^{-2})],$$
as $z\to\infty$. Substituted into \eq{lnD1},
\ben
&&\ln D_0(-z^2) = \int_a^b\left[s_{-1}(x)z+O(z^{-1})-\hf{u(x)}\right]dx-\half\ln(s_{-1}(a)z)-\half\ln(s_{-1}(b)z)\\
&&-O(z^{-2}) + \ln\left[1-\exp(-2\int_a^b(s_{-1}(x)z+O(z^{-1}))dx)\right]-\ln(2\phi_0(0;b))\\
&&=\int_a^b\left[s_{-1}(x)z-\hf{u(x)}\right]dx-\ln z-\hf{\ln s_{-1}(a)+\ln s_{-1}(b)}-\ln(2\phi_0(0;b))+O(z^{-1}).
\een
Substituted into \eq{zeta1ip}, as $s\to0$,
\ben
\zeta_1(s)&=&2s\frac{\sin\pi s}{\pi}\int_1^\infty z^{-2s-1}\ln D_0(-z^2)dz\\
 &=& 2s^2\int_1^\infty z^{-2s-1} \left\{z\int_a^b s_{-1}(x)dx-\ln z\right\}dz+O(s^2)\\
&&- 2s^2\int_1^\infty z^{-2s-1}dz\left\{\half\int_a^bu(x)dx+\hf{\ln s_{-1}(a)+\ln s_{-1}(b)}+\ln(2\phi_0(0;b))\right\}\\
&=& \frac{2s^2}{2s-1}\int_a^b\sqrt{1+f'^2}dx-\half+O(s^2)\\
&&- s\left\{\half\int_a^b u(x)dx+\hf{\ln s_{-1}(a)+\ln s_{-1}(b)}+\ln(2\phi_0(0;b))\right\}.
\een
This gives $\zeta_1(0)=-\half$, as in \eq{zeta12}, and furthermore the following theorem.
\begin{theorem}
The derivative of the zeta function in the one-dimensional case is given by
\be\label{zeta1p}
\zeta_1'(0) = -\hf{\ln f(a)+\ln f(b)} - \ln 2A,
\ee
where
\be\label{A}
A = \int_a^b\frac{\sqrt{1+f'(x)^2}}{f(x)}dx.
\ee
\end{theorem}
\begin{proof}
By the calculations above, we have
$$\zeta_1'(0)=-\left(\half\int_a^b u(x)dx+\hf{\ln s_{-1}(a)+\ln s_{-1}(b)}+\ln(2\phi_0(0;b))\right).$$
$\phi_0(0;b)$ can be computed exactly,
$$\phi_0(0;b) = \frac{f(a)}{\sqrt{1+f'(a)^2}}A,$$
where $A$ is defined by \eq{A}. Combined with $u(x)=(\ln f-\ln\sqrt{1+f'^2})'$, \eq{zeta1p} follows from a simplification.
\end{proof}

\section{Full zeta function and derivative}\label{2D}

For $k\neq0$, we make the change of variables $y=k^2z$. The differential equation for $\phi(x)$ becomes
\be\label{phi2}
\phi''(x) + u(x)\phi'(x) - k^2(1+f'^2)\left(z+\inv{f^2}\right)\phi(x) = 0,
\ee
where $u(x)$ is as defined in \sect{1D}. The zeta function in \eq{zeta2} becomes
\be\label{zeta2k}
\zeta_2(s)=2\frac{\sin\pi s}{\pi}\sum_{k=1}^\infty k^{-2s}\int_0^\infty z^{-s}\diff{\ln D_k(-k^2z)}{z}dz.
\ee
Because the integrand depends on $k$, the sum and the integral do not factor, but the sum nevertheless resembles the Riemann zeta function. This is a principal motivation for the change of variable. We define $S(k,z;x)$ as in \eq{Sdef}, and \eq{phi1S1} holds with $\phi_0(-z^2;b)$ replaced by $\phi_k(-k^2z;b)$. The equation for $S(k,z;x)$ is
$$ S'(k,z;x)+S(k,z;x)^2=\frac{1+f'^2}{f^2}(zf^2+1)k^2+\frac{u^2}{4}+\hf{u'},$$
and the solutions are written as asymptotic expansions in $k$,
\be\label{S1}
S_1(k,z;x)=s_{-1}(z;x)k+\frac{s_1(z;x)}{k}+\cdots.
\ee
The first two coefficients are
\be\label{sf2}
s_{-1}(z;x)=\frac{\sqrt{1+f'^2}}{f}\sqrt{t+1},\ \ \
s_1(z;x)=\frac{t}{(t+1)^{3/2}}\left[-\inv{8}\frac{t-4}{t+1}\frac{f'^2}{f\sqrt{1+f'^2}}+\inv{4}\frac{f''}{(1+f'^2)^{3/2}}\right],
\ee
where
\be\label{tdef}
t(x,z)=zf(x)^2.
\ee
Recall that $D_k(-k^2z)=\phi_k(-k^2z;b)/\phi_k(0;b)$. By \eq{lnD1},
\begin{eqnarray}\label{lnD2}
\ln D_k(-k^2z)&=&\int_a^b\left[S_1(k,z;x)-\hf{u(x)}\right]dx-\hf{\ln S_1(k,z;a)+\ln S_1(k,z;b)}\nonumber\\
& &+\ln\left[1-e^{-2\int_a^bS_1(k,z;x)dx}\right]-\ln(2\phi_k(0;b)).
\end{eqnarray}
Similarly, let
$$\diffp{\ln S_1(k,-k^2z;x)}{z}=f^2\diffp{}{t}\left[\ln(s_{-1}k)+\frac{s_1}{s_{-1}}k^{-2}+\cdots\right]
=f^2\left[t_0(z;x)+\frac{t_1(z;x)}{k^2}+\cdots\right],$$
where
\be\label{tf2}
t_0=\inv{2(1+t)},\ \ \ t_1=\frac{t^2-10t+4}{8(1+t)^4}\frac{f'^2}{1+f'^2}+\frac{1-t}{4(1+t)^3}\frac{ff''}{(1+f'^2)^2}.
\ee
Substituted into \eq{lnD2},
\ben
\frac \partial {\partial z}\ln D_k(-k^2z)&=&k\int_a^b f(x)^2\diffp{s_{-1}(z;x)}{t}dx+k^{-1}\int_a^b f(x)^2\diffp{s_1(z;x)}{t}dx+\cdots\\
&&-\hf{t_0(z;a)f(a)^2+t_0(z;b)f(b)^2}-\frac{t_1(z;a)f(a)^2+t_1(z;b)f(b)^2}{2k^2}-\cdots\\
&&+\frac{2e^{-2\int_a^b S_1(k,z;x)dx}}{1-e^{-2\int_a^b S_1(k,z;x)dx}}\diffp{}{z}\int_a^b S_1(k,z;x)dx.
\een
Substituted into \eq{zeta2k},
\ben
\zeta_2(s)&=&2\frac{\sin\pi s}{\pi}\zeta_R(2s-1)\int_a^b dx f(x)^{2s}\int_0^\infty t^{-s}\frac \partial {\partial t}s_{-1}(z;x)dt\\
&&+2\frac{\sin\pi s}{\pi}\zeta_R(2s+1)\int_a^b dx f(x)^{2s}\int_0^\infty t^{-s}\frac \partial {\partial t} s_1(z;x)dt+\cdots\\
&&-\frac{\sin\pi s}{\pi}\zeta_R(2s)\int_0^\infty[t_0(z;a)f(a)^{2s}t^{-s}+t_0(z;b)f(b)^{2s}t^{-s}]dt\\
&&-\frac{\sin\pi s}{\pi}\zeta_R(2s+2)\int_0^\infty[t_1(z;a)f(a)^{2s}t^{-s}+t_1(z;b)f(b)^{2s}t^{-s}]dt-\cdots\\
&&+2\frac{\sin\pi s}{\pi}\sum_{k=1}^\infty k^{-2s}\int_0^\infty dzz^{-s}\frac{2e^{-2\int_a^b S_1(k,z;x)dx}}
{1-e^{-2\int_a^b S_1(k,z;x)dx}}\diffp{}{z}\int_a^b S_1(k,z;x)dx.
\een
As before, $\zeta_R$ is the Riemann zeta function. Since $S_1(k,z;x)=O(k\sqrt{z})$ for large $k$ or $z$, the last term above does not contribute to the residues or values of
$\zeta_2(s)$ for $s$ being integers or half integers below (including) 1. Using the following standard formulas \cite{lebe72b},
$$\frac{\sin\pi s}{\pi}=\inv{\Gamma(s)\Gamma(1-s)},$$
and
$$\int_0^\infty \frac{z^{-s}}{(z+1)^p}dz=B(1-s,s+p-1)=\frac{\Gamma(1-s)\Gamma(s+p-1)}{\Gamma(p)},$$
we obtain
\ben
\zeta_2(s)&=&\zeta_R(2s-1)\frac{\Gamma(s-\half)}{\sqrt{\pi}\Gamma(s)}\int_a^bf(x)^{2s-1}\sqrt{1+f'(x)^2}dx\\
&&+\zeta_R(2s+1)\frac{\Gamma(s+\half)}{\sqrt{\pi}\Gamma(s)}s
\left[\frac{5s+1}{3}\int_a^b\frac{f(x)^{2s-1}f'(x)^2}{\sqrt{1+f'(x)^2}}dx+\int_a^b\frac{f(x)^{2s}f''(x)}{(1+f'(x)^2)^{3/2}}dx\right]+\cdots\\
&&-\half\zeta_R(2s)[f(a)^{2s}+f(b)^{2s}]\\
&&-\zeta_R(2s+2)f(a)^{2s}\frac{s^2}{4}\left[\frac{f'^2(a)}{1+f'^2(a)}\frac{5s+3}{4}+\frac{f''(a)f(a)}{(1+f'^2(a))^2}\right]\\
&&-\zeta_R(2s+2)f(b)^{2s}\frac{s^2}{4}\left[\frac{f'^2(b)}{1+f'^2(b)}\frac{5s+3}{4}+\frac{f''(b)f(b)}{(1+f'^2(b))^2}\right]-\cdots
\een
For $s=1,\half,0,-\half$, etc, the residue or value of $\zeta_2(s)$ may have two contributions from the series above, since both $\zeta_R(s)$ and $\Gamma(s)$ have poles.
For example, using properties of the Riemann zeta function, $\zeta_R(0)=-\half$ and $\res\,\,\zeta_R(1)=1$, we have
\be\label{zeta21}
\res\zeta_2(1)=\half\int_a^bf(x)\sqrt{1+f'(x)^2}dx,
\ee
\be\label{zeta22}
\res\zeta_2(\half)=-\inv{2\pi}\int_a^b\sqrt{1+f'(x)^2}dx-\frac{f(a)+f(b)}{4},
\ee
\be\label{zeta23}
\zeta_2(0)=\half,
\ee
\bes\label{zeta24}
&&\res\zeta_2(-\half)=-\inv{2\pi}\int_a^b\left[-\inv{8}\frac{f'^2}{f^2\sqrt{1+f'^2}}+\inv{4}\frac{f''}{f(1+f'^2)^{3/2}}\right]dx\non\\
&&-\inv{32f(a)}[\inv{8}\frac{f'^2(a)}{1+f'^2(a)}+\frac{f''(a)f(a)}{(1+f'^2(a))^2}]
-\inv{32f(b)}\left[\inv{8}\frac{f'^2(b)}{1+f'^2(b)}+\frac{f''(b)f(b)}{(1+f'^2(b))^2}\right].
\ees
Adding up $\zeta_1(s)$ and $\zeta_2(s)$, and recalling that $\zeta_1(1)$ is finite, we obtain the following theorem for the full zeta function.
\begin{theorem}\label{zetaf}
For the full zeta function of the Laplacian on a surface of revolution,
\be\label{zetaf1}
\res\zeta(1)=\half\int_a^bf(x)\sqrt{1+f'(x)^2}dx.
\ee
\be\label{zetaf2}
\res\zeta\left(\half \right)=-\frac{f(a)+f(b)}{4}.
\ee
\be\label{zetaf3}
\zeta(0)=0.
\ee
\be\label{zetaf4}
\res\zeta \left(-\half\right)=-\inv{32f(a)}\left[\inv{8}\frac{f'^2(a)}{1+f'^2(a)}+\frac{f''(a)f(a)}{(1+f'^2(a))^2}\right]
-\inv{32f(b)}\left[\inv{8}\frac{f'^2(b)}{1+f'^2(b)}+\frac{f''(b)f(b)}{(1+f'^2(b))^2}\right].
\ee
\end{theorem}
The geometric interpretation of the formulas above will be discussed in \sect{discuss}.

\vspace{0.2in}{\bf Derivative of $\zeta_2 (s)$ at $s=0$}

The derivative of $\zeta_2(s)$ at 0 can be calculated in a similar but more subtle way. The following theorem is the main result of the paper.
\begin{theorem}
The derivative at zero of the holomorphic extension of the zeta function on the two-dimensional surface of revolution is given by
\be\label{zetafp}
\zeta'(0) = -2\ln\phi\left(e^{-2\frac{\pi^2}{A}}\right) + \frac{\pi^2}{6A} + \inv{6}\int_a^b\frac{f'(x)^2}{f(x)\sqrt{1+f'(x)^2}}dx
 + \half\int_a^b\frac{f''(x)}{(1+f'(x)^2)^{3/2}} dx,
\ee
where $A$ is defined in \eq{A}, and $\phi$ is the Euler function \cite{apos76b,NIST} defined by
\be\label{euler}
\phi(q) = \prod_{k=1}^\infty (1-q^k).
\ee
\end{theorem}
\begin{proof}
Integrating \eq{zeta2k} by parts, we find that
\be\label{zeta2ip}
\zeta_2(s)=2s\frac{\sin\pi s}{\pi}\sum_{k=1}^\infty k^{-2s}\int_0^\infty z^{-s-1}\ln D_k(-k^2z)dz.
\ee
$S_1(k,0;x)$ and $\phi_k(0;b)$ can be computed exactly,
$$S_1(k,0;x)=\frac{\sqrt{1+f'(x)^2}}{f(x)}k,\ \ \ \phi_k(0;b)=\frac{f(a)}{\sqrt{1+f'(a)^2}}\frac{e^{kA}-e^{-kA}}{2k},$$
where $A$ is defined by \eq{A}. This allows us to write
\be\label{lnD2p}
\ln D_k(-k^2z) = \int_a^b S_1(k,z;x)dx-Ak -\half\ln\frac{S_1(k,z;a)S_1(k,z;b)}{S_1(k,0;a)S_1(k,0;b)}
+ \ln\frac{1-e^{-2\int_a^b S_1(k,z;x)dx}}{1-e^{-2Ak}}.
\ee
By the asymptotic expansion in \eq{S1}, we write $S_1(k,z;x)$ as
$$ S_1(k,z;x)=s_{-1}(z;x)k+\frac{s_1(z;x)}{k}+R(k,z;x), $$
where the remainder $R(k,z;x)=O(y^{-3/2})=O(k^{-3}z^{-3/2})$ as $y\to\infty$. In addition, since $S_1(k,0;x)=s_{-1}(0;x)k$, we have $s_1(z;x)=O(z)$ and $R(k,z;x)=O(k^{-3}z)$ as $z\to 0$. It implies that $\ln S_1(k,z;x)$ can be written as
$$ \ln S_1(k,z;x) = \ln(s_{-1}(z;x)k) + D(k,z;x), $$
where $D(k,z;x)=O(y^{-1})=O(k^{-2}z^{-1})$ as $y\to\infty$, and $D(k,z;x)=O(k^{-2}z)$ as $z\to 0$. Substituting \eq{lnD2p} into \eq{zeta2ip} and using the expressions for $S_1(k,z;x)$, we have
\ben
&&\!\!\!\!\!\!\!\!\zeta_2(s)=2s\frac{\sin\pi s}{\pi}\sum_{k=1}^\infty k^{-2s}\int_0^\infty z^{-s-1}\left\{\int_a^b
\left[\frac{\sqrt{1+f'^2}}{f}k(\sqrt{t+1}-1)+\frac{s_1(z;x)}{k}+R(k,z;x)\right]dx\right.\\
&&\left.-\frac{\ln(zf(a)^2+1)+\ln(zf(b)^2+1)}{4}-\hf{D(k,z;a)+D(k,z;b)}+\ln\frac{1-e^{-2\int_a^b S_1(k,z;x)dx}}{1-e^{-2Ak}}\right\}dz.
\een
The asymptotics of $R(k,z;x)$ and $D(k,z;x)$ as $y\to\infty$ and as $z\to0$ implies that
$$\sum_{k=1}^\infty k^{-2s}\int_0^\infty z^{-s-1}dz\left\{\int_a^b R(k,z;x)dx-\hf{D(k,z;a)+D(k,z;b)}\right\}$$
is analytic for $-\frac 1 2 <\Re s < 1$. Therefore,
\ben
\zeta_2(s)&=&2s\frac{\sin\pi s}{\pi}\sum_{k=1}^\infty k^{-2s}\int_0^\infty z^{-s-1}\left\{\int_a^b
\left[\frac{\sqrt{1+f'^2}}{f}k(\sqrt{t+1}-1)+\frac{s_1(z;x)}{k}\right]dx\right.\\
&&\left.-\frac{\ln(zf(a)^2+1)+\ln(zf(b)^2+1)}{4}+\ln\frac{1-e^{-2\int_a^b S_1(k,z;x)dx}}{1-e^{-2Ak}}\right\}dz+O(s^2).
\een
We now calculate each one of these terms individually. For the first one, we have
\ben
&&2s\frac{\sin\pi s}{\pi}\sum_{k=1}^\infty k^{-2s}\int_0^\infty z^{-s-1}dz\intab\frac{\sqrt{1+f'(x)^2}}{f(x)}k(\sqrt{t+1}-1)dx\\
&=&2\frac{\sin\pi s}{\pi}\sum_{k=1}^\infty k^{1-2s}\intab f(x)^{2s}\frac{\sqrt{1+f'(x)^2}}{f(x)}dx\int_0^\infty \frac{t^{-s}}{2\sqrt{t+1}}dt\\
&=&\zeta_R(2s-1)\frac{\Gamma(s-\half)}{\Gamma(\half)\Gamma(s)}\intab f(x)^{2s-1}\sqrt{1+f'(x)^2}dx\\
&=&s\frac{A}{6}+O(s^2),
\een
where we used $\zeta_R(-1)=-1/12$.

For the second term, substituting $s_1(z;x)$ from \eq{sf2},
\ben
&&\!\!2s\frac{\sin\pi s}{\pi}\sum_{k=1}^\infty k^{-2s}\int_0^\infty z^{-s-1}dz\intab\frac{s_1(z;x)}{k}dx\\
&=&\!\!2s\frac{\sin\pi s}{\pi}\sum_{k=1}^\infty k^{-2s-1}\intab\frac{f^{2s}(x)dx}{\sqrt{1+f'(x)^2}}
\int_0^\infty \frac{t^{-s-1}dt}{\sqrt{t+1}}\left[-\frac{t(t-4)}{8(t+1)^2}\frac{f'(x)^2}{f(x)}+\inv{4}\frac{t}{t+1}\frac{f''(x)}{1+f'(x)^2}\right]\\
&=&\!\!s\zeta_R(2s+1)\frac{\Gamma(s+\half)}{\Gamma(\half)\Gamma(s)}\left[\frac{5s+1}{3}
\intab\frac{f(x)^{2s}f'(x)^2}{f(x)\sqrt{1+f'(x)^2}}dx+\intab\frac{f(x)^{2s}f''(x)}{(1+f'(x)^2)^{3/2}}dx\right]\\
&=&\!\!s\left[\inv{6}\intab\frac{f'(x)^2}{f(x)\sqrt{1+f'(x)^2}}dx+\half\intab\frac{f''(x)}{(1+f'(x)^2)^{3/2}}dx\right]+O(s^2).
\een
For the third term, we obtain
\ben
&&2s\frac{\sin\pi s}{\pi}\sum_{k=1}^\infty k^{-2s}\int_0^\infty z^{-s-1}\left[-\frac{\ln(zf(a)^2+1)+\ln(zf(b)^2+1)}{4}\right]dz\\
&=&-\half\frac{\sin\pi s}{\pi}\sum_{k=1}^\infty k^{-2s}\int_0^\infty \frac{t^{-s}dt}{t+1}[f(a)^{2s}+f(b)^{2s}]\\
&=&-\half\zeta_R(2s)[f(a)^{2s}+f(b)^{2s}]\\
&=&\half+s\left[-2 \zeta'_R(0)+\hf{\ln f(a)+\ln f(b)}\right]+O(s^2)\\
&=&\half+s\left[\ln(2\pi)+\hf{\ln f(a)+\ln f(b)}\right]+O(s^2),
\een
where we used $\zeta_R'(0)=-\frac 1 2 \ln(2\pi)$. The fourth term involves the Euler function defined in \eq{euler},
\ben
&&2s\frac{\sin\pi s}{\pi}\sum_{k=1}^\infty k^{-2s}\int_0^\infty z^{-s-1}\ln\frac{1-e^{-2\int_{a}^{b}S_1(k,z;x)dx}}{1-e^{-2Ak}}dz\\
&=&2s\frac{\sin\pi s}{\pi}\left\{\sum_{k=1}^\infty k^{-2s}\int_1^\infty z^{-s-1}\ln\frac{1-e^{-2\intab S_1(k,z;x)dx}}{1-e^{-2Ak}}dz+O(1)\right\}\\
&=&-2s\frac{\sin\pi s}{\pi}\left\{\sum_{k=1}^\infty k^{-2s}\int_1^\infty z^{-s-1}\ln(1-e^{-2Ak})dz+O(1)\right\}\\
&=&-2s\sum_{k=1}^\infty\ln(1-e^{-2Ak})+O(s^2)\\
&=&-2s\ln\phi(e^{-2A})+O(s^2).
\een
Adding these four terms gives
\bes\label{zeta2p1}
\zeta_2'(0)&=&-2\ln\phi(e^{-2A})+\frac{A}{6}+\ln(2\pi)+\hf{\ln f(a)+\ln f(b)}\non\\
&&+\inv{6}\intab\frac{f'(x)^2}{f(x)\sqrt{1+f'(x)^2}}dx+\half\intab\frac{f''(x)}{(1+f'(x)^2)^{3/2}}dx.
\ees
Applying the following formula, which will be proved as \lem{lemma},
$$ \ln\phi(e^{-2\pi r})-\ln\phi\left(e^{-\frac{2\pi}{r}}\right)=\half\left[\frac{\pi}{6}(r-\inv{r})-\ln r\right], $$
with $r=A/\pi$, we have
\bes\label{zeta2p2}
\zeta_2'(0)&=&-2\ln\phi(e^{-2\frac{\pi^2}{A}})+\frac{\pi^2}{6A}+\ln(2A)+\hf{\ln f(a)+\ln f(b)}\non\\
&&+\inv{6}\intab\frac{f'(x)^2}{f(x)\sqrt{1+f'(x)^2}}dx+\half\intab\frac{f''(x)}{(1+f'(x)^2)^{3/2}}dx.
\ees
We now add $\zeta_1'(0)$ from \eq{zeta1p} with $\zeta_2'(0)$ above to obtain \eq{zetafp}.
\end{proof}

\begin{lemma}\label{lemma}
The following formula holds for the Euler function,
\be\label{Euler}
\ln\phi \left(e^{-2\pi r}\right)-\ln\phi \left(e^{-\frac{2\pi}{r}}\right)=\half\left[\frac{\pi}{6}(r-\inv{r})-\ln r\right].
\ee
\end{lemma}
\begin{proof}
Consider the special case of \( f\equiv R/\pi \) on an interval of length $L.$ Then the eigenvalues are
$$\lambda_{k,n}=\left(\frac{k\pi}{R}\right)^2+\left(\frac{n\pi}{L}\right)^2.$$
Due to the symmetry between $R$ and $L$ in the spectrum, $\zeta_2(s)(R,L)=\zeta_2(s)(L,R)$. Therefore, by \eq{zeta2p1},
$$\zeta_2'(0)=-2\ln\phi \left(e^{-2\pi\frac{L}{R}}\right)+\frac{\pi}{6}\frac{L}{R}+\ln(2R)
=-2\ln\phi \left(e^{-2\pi\frac{R}{L}}\right)+\frac{\pi}{6}\frac{R}{L}+\ln2L.$$
Letting $r=L/R$, we obtain \eq{Euler}.
\end{proof}

\section{Discussion and conclusion}\label{discuss}

The geometric interpretation of the zeta function's residues and values can be found from the heat kernel expansion of the Laplacian on the surface of revolution;
Refs. \cite{gilk95b,kirs02b,vass03-388-279} gave extensive reviews of the heat kernel expansion. On a compact manifold, the heat kernel of the Laplacian, integrated over space, is defined as
$$ \theta(t) = \sum_{\lambda_k \in {\rm Spec}(\Delta)} e^{-t\lambda_k}. $$
For a two-dimensional manifold with boundary, denoted by $\Omega$, the heat kernel has the following expansion,
$$\theta(t)=\frac{C_{-1}}{t}+\frac{C_{-\half}}{\sqrt{t}}+C_0+C_\half\sqrt{t}+C_1t+\cdots.$$
The coefficients are integrations of geometric invariants over $\Omega$ and its boundary $\partial\Omega$ \cite{bran90-15-245}. For the surface of revolution, if the Dirichlet boundary condition is taken,
$$C_{-1}=\frac{|\Omega|}{4\pi},\ \ \ C_{-\half}=-\frac{|\partial\Omega|}{8\sqrt{\pi}},
\ \ \ C_0=\frac{\chi}{6},\ \ \ C_\half=\inv{256\sqrt{\pi}}\int_{\partial\Omega}[k_g^2(\sigma)-8k(\sigma)]d\sigma,$$
where $|\Omega|$ is the area of $\Omega$, $|\partial\Omega|$ is the total length of $\partial\Omega$, $\chi$ is Euler characteristic of $\Omega$, $k_g(\sigma)$ is the geodesic curvature of the boundary at $\sigma$, and $k(\sigma)$ is the Gaussian curvature of the surface at $\sigma$. For the surface of revolution, $\chi=0$.

The zeta function is related to the heat kernel by
$$\zeta(s)=\inv{\Gamma(s)}\int_0^\infty dt\theta(t)t^{s-1}.$$
Substituting the heat kernel expansion into the equation above, we have
\ben
\res\zeta(1) &=& C_{-1}=\frac{|\Omega|}{4\pi}=\inv{4\pi}\int_a^b 2\pi f(x)\sqrt{1+f'(x)^2}dx,\\
\res\zeta(\half) &=& \frac{C_{-\half}}{\sqrt{\pi}}=-\frac{|\partial\Omega|}{8\pi}=-\frac{2\pi f(a)+2\pi f(b)}{8\pi},\\
\zeta(0) &=& C_0=\frac{\chi}{6}=0,\\
\res\zeta(-\half) &=& -\frac{C_\half}{2\sqrt{\pi}}=-\inv{512\pi}\int_{\partial\Omega}[k_g^2(\sigma)-8k(\sigma)]d\sigma\\
&=&-\frac{2\pi f(a)}{512\pi}\left[\frac{f'^2(a)}{f^2(a)(1+f'^2(a))}+8\frac{f''(a)}{f(a)(1+f'^2(a))^2}\right]\\
&&-\frac{2\pi f(b)}{512\pi}\left[\frac{f'^2(b)}{f^2(b)(1+f'^2(b))}+8\frac{f''(b)}{f(b)(1+f'^2(b))^2}\right].
\een
They agree with \thm{zetaf} completely. It confirms the validity of the contour integral method extended to the surface of revolution. On the other hand,
the determinant of the Laplacian, calculated by the contour integral method, cannot be derived from the heat kernel expansion directly.

Our results can also be benchmarked against an exactly solvable model with continuous spectrum. Consider the Laplacian on an infinitely long planar stripe of width $L$.
Given the infinite volume of this configuration we need to introduce suitable zeta function densities. To that end, we compactify the infinitely long planar stripe temporarily to a
a finite stripe of width $2\pi R$ and impose periodic boundary conditions along that direction.
With Dirichlet boundary condition along the remaining direction, the spectrum is
$$\lambda_{k,n}=\left(\frac {k} R\right)^2 +\left(\frac{n\pi}{L}\right)^2,\ \ \ k\in\ZZ,\ \ \ n=1,2,3,\ldots.$$
The associated zeta function is
\begin{eqnarray}
\zeta (s) - \sum_{n=1}^\infty \sum_{k=-\infty}^\infty \left[ \left( \frac{ k} R \right)^2 + \left( \frac{n\pi} L \right)^2 \right]^{-s}.\nonumber
\end{eqnarray}
As $R\to\infty$, which is the transition to the Riemann integral, the relevant zeta function density is
\ben
\zeta_c(s)&:=&\lim_{R\to\infty} \frac{1} R \zeta (s) = \sum_{n=1}^\infty\int_{-\infty}^\infty\lambda_{nk}^{-s}dk\\
&=&2\zeta_R(2s-1)\int_0^\infty(x^2+1)^{-s}dx\left(\frac{L}{\pi}\right)^{2s-1}\\
&=&\zeta_R(2s-1)B(\half,s-\half)\left(\frac{L}{\pi}\right)^{2s-1}\\
&=&\sqrt{\pi}\zeta_R(2s-1)\frac{\Gamma(s-\half)}{\Gamma(s)}\left(\frac{L}{\pi}\right)^{2s-1}.
\een
In particular, we have
$$\res\zeta_c(1)=\hf{L},\ \ \res\zeta_c \left(\half \right)=-\half,\ \ \zeta_c(0)=0,\ \ \res\zeta_c \left(-\half \right)=0,$$
and
$$\zeta'_c(0)=\frac{\pi^2}{6L}.$$
The stripe can be regarded as a surface of revolution with $f(x)\equiv R$, and $R\to\infty$. For $f(x)\equiv R$, \thm{zetaf} gives
$$ \res\zeta(1)=\hf{LR},\ \ \res\zeta \left(\half \right)=-\hf{R},\ \ \zeta(0)=0,\ \ \res\zeta \left(-\half \right)=0,$$
and
$$ \zeta'(0) = -2\ln\phi \left(e^{-2\frac{\pi^2R}{L}}\right) + \frac{\pi^2R}{6L}. $$
One easily verifies, as the above computation shows, that these properties are recovered from
$$\zeta_c(s)=\lim_{R\to\infty}\frac{ \zeta (s)}{R}.$$

As a conclusion, we have extended the contour integral method to the surface of revolution. By the WKB expansion, we calculated the zeta function at several important points. The results agree with those obtained from the heat kernel expansion. We also computed the determinant of the Laplacian on such a surface. The WKB expansion allows the evaluation of the zeta function at more points,
however the expressions get increasingly more complicated. In the future, we can change the Dirichlet boundary condition to, say,
Robin boundary condition, and calculate the associated zeta function. We can also relax the condition that $f(x)>0$ by allowing $f(a)=0$ and/or $f(b)=0$.
This should allow the study of manifolds with cusps.\\[.2cm]

{\bf Acknowledgments}\\[.2cm]
KK is supported by the National Science Foundation Grant PHY-0757791.

\end{document}